\documentclass{article}

\usepackage{stylefile}

\usepackage[utf8]{inputenc} 
\usepackage[T1]{fontenc}    
\usepackage{hyperref}       
\usepackage{url}            
\usepackage{booktabs}       
\usepackage{amsfonts}       
\usepackage{nicefrac}       
\usepackage{microtype}      
\usepackage{graphicx}
\usepackage{hyphenat}
\graphicspath{ {./images/} }

\usepackage{amsthm}
\newtheorem{theorem}{Theorem}[section]

\newtheorem{prop}[theorem]{Proposition}

\newtheorem*{remark}{Remark}

\usepackage[linesnumbered,ruled,vlined]{algorithm2e}
\SetKwInput{KwInput}{Input}                
\SetKwInput{KwOutput}{Output}  
\usepackage{mathtools}
\usepackage{bm}
\DeclarePairedDelimiter\floor{\lfloor}{\rfloor}

\title{Some variations of the secretary problem}

\author{
 Sarthak Agrawal \\
 Computer Science and Engineering\\
  Indian Institute Of Technology Kanpur\\
   \texttt{asarthak@iitk.ac.in} \\
    \And
 Sanjeev Saxena \\
 Computer Science and Engineering\\
  Indian Institute Of Technology Kanpur\\
  \texttt{ssax@iitk.ac.in}\\
}

\begin{document}
\maketitle
\begin{abstract}
We consider two variations of the classical secretary problem. 
\begin{itemize}
    \item A variation of the returning secretary problem where each interviewee may appear a second time with a fixed probability $p$. The decision-maker observes interviewees sequentially and must choose whether to accept or reject each appearance. We characterize the optimal threshold rule and examine its dependence on the reappearance probability $p$, highlighting how additional information from repeated appearances improves selection performance.
    \item A variation of the secretary problem in which success is defined as selecting any one of the top three interviewees rather than the single best. Interviewees are observed sequentially in random order, and decisions are irreversible. We estimated the success probability under this relaxed success criterion using the threshold strategy of the classical secretary problem. The results show that allowing selection among the top three significantly increases the success probability and shifts the optimal stopping threshold earlier than in the classical problem. This model provides insight into realistic decision-making scenarios where top interviewees are more or less similar. 
\end{itemize}
\end{abstract}

\section{Introduction}

The secretary problem is a problem in optimal stopping
theory\cite{Hill2009,SecProbWiki} and also a classic example of the
random-arrival model\cite{vardi:LIPIcs:2015:4953}. In this problem,
there are $n$ interviewees for the secretary position. These
interviewees arrive in random order and are interviewed sequentially.
We know (or can remember) the best interviewee interviewed so far.
However, the employer is unaware of the interviewees he has not yet
seen. Moreover, immediately after the interview, the employer either
has to select or reject the interviewee. An interviewee once rejected
cannot be reconsidered. We are to determine a strategy such that the
best interviewee is selected with maximum
probability\cite{SecProbWiki}. 

This problem has been studied extensively in the fields of applied
probability, statistics, and decision theory. 
This problem has been extensively studied and the best shortest proof
so far is 
by the odds algorithm\cite{bruss2000}. We basically reject
approximately the first $n/e$ (cutoff value) interviewees and select
the first one who is better than those already seen\cite{dynkin}. 

This classical problem has many
variants, such as:
\begin{itemize}
    \item K-choice secretary problem\cite{Babaioff2008}.
    \item The knapsack secretary problem\cite{Babaioff2007}.
    \item The matroid secretary problem\cite{babaioff2007matroids}.
    \item Returning secretary problem\cite{Garrod2012jan,Vardi2014,Ribas2018}.
    \item The best or worst problem\cite{Bayn2017}
\end{itemize}

\subsection{Problem Statement}

Let $N$ distinct interviewees be totally ordered by a strict ranking,
with ranks $\{1,2,\dots,N\}$, where rank $1$ is the best. Interviewees
arrive randomly over time according to a uniform distribution.

Each interviewee arrives at least once. At each appearance, the
decision-maker observes only the relative rank of the appearing
interviewee among all interviewees observed so far (across all
appearances), but not the interviewee's absolute rank nor whether the
interviewee will reappear again in the future.

After observing an appearance, the decision-maker must immediately
either \textit{accept} the interviewee (terminating the process) or
\textit{reject} the appearance. Rejection of an interviewee at his/her
last appearance permanently eliminates that interviewee from future
consideration. The goal is to maximize the probability of selecting
the rank one (the best) interviewee.

A stopping rule is a (possibly randomized) policy generated by the
observed relative ranks. The problem is to characterize an optimal
stopping rule and the corresponding maximum success probability.

\subsection{Assumptions and terminologies}
Let us define some terminology, 
\begin{itemize}
\item An interviewee who is better than all the interviewees who have
been interviewed so far is called \textit{leading}.

\item Interviewee selected following the proposed strategy is called
\textit{chosen}.
\end{itemize}

\section{Technical results} \label{chap:tech-results} 

Here, we summarize some technical results that we have used for
analysis throughout the paper. These results are mainly from Ribas
\cite{Ribas2018}.
\begin{quote}
    
\begin{theorem}
\label{thm:1}
    Ribas[Proposition-1]\cite{Ribas2018}:
    Let $\{ F_n\}_{n\in \mathbb{N}}$ be a sequence of real functions
    with $F_n$ defined by integer values in $\{0,...,n\}$ and let $M(n)$ be the value for which the function $F_n$ reaches its maximum. Assume that the sequence of functions $\{ f_n\}_{n\in\mathbb{N}}$ defined by $f_n(x) := F_n(\floor{nx})$, converges uniformly on $[0,1]$ to $f$ continuous in $[0,1]$ and that $\theta$ is the only global maximum of $f$ in $[0,1]$. Then,
    \begin{enumerate}
        \item $\underset{n}{\lim} \frac{M(n)}{n} = \theta$
        \item $\underset{n}{\lim} F_n(M(n)) = f(\theta)$
        \item If $\mathfrak{M} \sim M(n)$ then $\underset{n}{\lim} F_n(\mathfrak{M}(n)) = f(\theta)$ 
    \end{enumerate}
\end{theorem}

\begin{theorem}
    \label{thm:2}
    Ribas[Theorem-1]\cite{Ribas2018}:
    Consider the sequences of functions $\{F_n\}_{n\in\mathbb{N}}$, $\{G_n\}_{n\in\mathbb{N}}$ and $\{H_n\}_{n\in\mathbb{N}}$ with $F_n, G_n , H_n : [0,n] \cap \mathbb{Z}\rightarrow \mathbb{R}$ and defined recursively by the conditions:
    \begin{equation*}
        F_n(k) = G_n(k) + H_n(k)F_n(k-1) \ \ and \ \ F_n(0) = \mu
    \end{equation*}
    Also for all $x$, $0 \leq x \leq 1$, Let $f_n(x) := F_n(\floor{nx})$, $h_n(x) := n(1-H_n(\floor{nx}))$ and $g_n(x) := nG_n(\floor{nx})$. If the following conditions hold:
    \begin{itemize}
        \item Both $h_n(x)$ and $g_n(x)$ converges on $(0,1)$ and uniformly on a $[\epsilon,\epsilon^{\prime}]$ for all $0<\epsilon<\epsilon^{\prime}<1$ to continuous functions in $(0,1)$, $h(x)$ and $g(x)$, respectively.
        \item $f_n(x)$ converges uniformly in $[0,1]$ to a continuous function $f(x)$.
    \end{itemize}
    Then, $f(0) = \mu$ and $f$ satisfy in $(0,1)$,
    \begin{equation*}
        f^\prime(x) = -f(x)h(x) + g(x)
    \end{equation*}
\end{theorem}

\begin{theorem}
    \label{thm:3}
    Ribas[Theorem-2]\cite{Ribas2018}:
    Consider the sequences of functions $\{F_n\}_{n\in\mathbb{N}}$, $\{G_n\}_{n\in\mathbb{N}}$ and $\{H_n\}_{n\in\mathbb{N}}$ with $F_n, G_n , H_n : [0,n] \cap \mathbb{Z}\rightarrow \mathbb{R}$ and defined recursively by the conditions:
    \begin{equation*}
        F_n(k) = G_n(k) + H_n(k)F_n(k+1) \ \ and \ \ F_n(n) = \mu
    \end{equation*}
    Also for all $x$, $0 \leq x \leq 1$, Let $f_n(x) := F_n(\floor{nx})$, $h_n(x) := n(1-H_n(\floor{nx}))$ and $g_n(x) := nG_n(\floor{nx})$. If the following conditions hold:
    \begin{itemize}
        \item Both $h_n(x)$ and $g_n(x)$ converges on $(0,1)$ and uniformly on a $[\epsilon,\epsilon^{\prime}]$ for all $0<\epsilon<\epsilon^{\prime}<1$ to continuous functions in $(0,1)$, $h(x)$ and $g(x)$, respectively.
        \item $f_n(x)$ converges uniformly in $[0,1]$ to a continuous function $f(x)$.
    \end{itemize}
    Then, $f(1) = \mu$ and $f$ satisfy in $(0,1)$,
    \begin{equation*}
        f^\prime(x) = f(x)h(x) - g(x)
    \end{equation*}
\end{theorem}
\end{quote}

Here $f^\prime(x)$ is the derivative of $f(x)$.

\section{Secretary Problem with probabilistic second arrival}
\label{chap:relatedchap1}

In this case, each interviewee arrives at least once. After his/her
first appearance, interviewee $i$ may appear a second time with
probability $p \in [0,1]$. No interviewee arrives more than twice.

Note that:
\begin{enumerate} \label{notedresults}
\item If $k$ distinct interviewees are interviewed, and the employer
has to select the leading interviewee from them, then the probability
that the chosen interviewee is the overall best is $\frac{k}{n}$, as
in the secretary problem\cite{Ribas2018}.

\item If the interviewee is guaranteed to appear a second time, it is
always preferable for an employer not to accept an interviewee on his
first arrival, as other better interviewees might appear before
selection, which will increase the probability of success.

\item When an interviewee who is inferior to the best interviewee
interviewed so far arrives for the interview, it is irrelevant for the
employer to interview him. Thus, he/she can be rejected directly
\cite{Ribas2018}.

\item The only relevant information to select the overall best
interviewee with maximum probability is the number of distinct
interviewees has been interviewed so far and the number of times the
interviewer has interviewed the leading interviewee.

\item Therefore, after each interview, the employer can directly
reject all interviewees who have already been interviewed and are
inferior to the best interviewee so far \cite{Ribas2018}.  
\end{enumerate}
\begin{remark}
    \upshape
    \label{assumption}
Instead of rejecting inferior interviewees, we can model this by
saying that all his/her appearances have already been interviewed. For
example, suppose an employer has interviewed $k$ distinct
interviewees, in which the leading interviewee has appeared only once.
In that case, we can consider that all non-leading interviewees in
$k$, i.e., $(k-1)$, have exhausted all their appearances. So now the
employer has taken approximately $(p+1)(k-1)+1$ interviews from
$(1+p)n$ interviews.
\end{remark}

\subsection{Proposed strategy}

In {Observation Phase}, we interview $k$ different interviewees,
reject everyone but note down the leading among them, and the number
of times he/she has appeared for the interview.

Then comes the {Selection Phase}. After interviewing $k$ distinct
interviewees, select the first interviewee that satisfies any of the
following criteria in sequential order:
\begin{enumerate}
\item 
S/he is the leading interviewee and s/he arrives again.

\item 
A better (better interviewee than all the previously seen
interviewees) interviewee arrives for the first time, select him/her
with probability $(1-p)$. Note that he/she can appear again with
probability $p$.

\item 
A better (better interviewee than all the previously seen
interviewees) interviewee arrives a second time, select him/her.
\end{enumerate}

\subsection{Probability of success for probabilistic second arrival}

The probability of success following the proposed strategy is the
probability that the chosen interviewee is the overall best, assuming
we have rejected $k$ distinct interviewees.\\
The probability of success can be calculated by combining the
probability that the chosen interviewee is overall best, with respect
to two mutually exclusive and exhaustive events:
\begin{itemize}
\item $\chi_1$: Chosen interviewee is the overall best, assuming we
have rejected $k$ distinct interviewees and the leading interviewee
has appeared only once.

\item $\chi_2$: Chosen interviewee is the overall best, assuming we
have rejected $k$ distinct interviewees and the leading interviewee
has appeared twice.
\end{itemize}
Let us now introduce some notations. This is similar to that in
\cite{Ribas2018}. Let
\begin{itemize}
\item $\Phi_n(k)$: Probability that the chosen interviewee is overall
best, assuming we have rejected $k$ distinct interviewees and the
leading interviewee has appeared only once. This is the probability of
event $\chi_1$.

\item $\Psi_n(k)$: Probability that the chosen interviewee is overall
best, assuming we have rejected $k$ distinct interviewees and the
leading interviewee has appeared twice. This is the probability of
event $\chi_2$.

\item $\Upsilon_n(k)$: Probability that the leading interviewee has
appeared only once when the $k^{th}$ distinct interviewee makes
his/her initial appearance.

\item $F_n(k)$: Probability that the chosen interviewee is overall
best, assuming we have rejected $k$ distinct interviewees.
\end{itemize}

\subsubsection{Calculation of $\Phi_n(k)$}

We calculate the value of $\Phi_n(k)$ as follows: 
\begin{prop}
\label{prop:probArrival.1} For all natural numbers $n$ and $k$ where
$k < n$, we have,
\begin{equation*}
    \begin{split}
        &\Phi_n(k) = \frac{pak + (1-p)(1-pa)}{n} + \frac{(p+k)(1-pa)}{k+1}\Phi_n(k+1), \\
        &where \ \ a = \frac{1}{(1+p)(n-k)+1}
    \end{split}
\end{equation*}
and $\Phi_n(n) = p$.
\end{prop}
\begin{proof}
By hypothesis, when all $n$ distinct interviewees have been
interviewed, then leading is overall best, and he/she has only
appeared once, so he/she may appear again with probability $p$. So
$\Phi_n(n)=p$,\\
Now, for the 
case when $k<n$, 
the remaining interviews can be categorized into two mutually
exclusive events:
\begin{enumerate}
\item \textit{$E$: The leading interviewee will appear again}\\
As each interviewee can appear a second time with probability $p$, the
probability that the leading interviewee will appear again is \[P(E) =
p\] Let $P(S_1)$ be $\Phi_n(k)|E$ i.e, probability obtained with
respect to event $E$.\\ %
For the next interviewee that arrives, the following scenarios are
possible:

\begin{description}
\item[Case-1:] \textit{$X_1$: Next interviewee is the second
appearance of the leading interviewee.}
\begin{itemize}
\item As we have observed $k$ distinct interviewees out of $n$
distinct interviewees, the number of remaining distinct interviewees
is $n-k$. Since each of these interviewees can appear $(1+p)$ times,
and the second appearance of the current leading interviewee is also
yet to occur, the total number of interviewees yet to be seen is
$(1+p)(n-k)+1$. So, the probability that the next interviewee is the
second appearance of the leading interviewee is \[P(X_1) =
\frac{1}{(1+p)(n-k)+1}\]

\item And according to the proposed strategy, this interviewee will
be our chosen interviewee, so the probability that the chosen
interviewee is overall best in this case is $\frac{k}{n}$
[Page-\pageref{notedresults}, 
Point-1][\ref{notedresults}]. Therefore, \[P(S_1|X_1) = \frac{k}{n}\]
\end{itemize}
        
\item[Case-2:] \textit{$X_2$: Next interviewee is a new leading
interviewee.}
\begin{itemize}
\item For the next interviewee to be a new leading interviewee, the
following two events should happen together:
\begin{enumerate}
\item Next interviewee should not be the second appearance of the
leading interviewee, which happens with probability $(1-P(X_1))$.

\item Next interviewee should be the new leading interviewee. As the
leading interviewee will be only one among the $k+1$ distinct
interviewees interviewed so far, the probability of this happening is
$(\frac{1}{k+1})$
\end{enumerate}

The total probability is the product of the probabilities of both
events, i.e. \[P(X_2) =
\left(1-\frac{1}{(1+p)(n-k)+1}\right)\frac{1}{k+1}\]

\item And according to the proposed strategy, we will accept this
interviewee with probability $1-p$ and reject him/her with probability
$p$. Once rejected, the interview process will continue%
(but we have seen $k+1$ different candidates). So the probability
$P(S_1)$ with respect to event $X_2$ is. 
\[P(S_1|X_2) = \left(1-p\right)\left(\frac{k+1}{n}\right) +
p\left(\Phi_n(k+1)\right)\]
\end{itemize}

\item[Case-3:] \textit{$X_3$: Next interviewee is a new non-leading
interviewee.}
\begin{itemize}
\item For the next interviewee to be a new non-leading interviewee,
the following two events should happen together:
\begin{enumerate}
\item Next interviewee should not be the second appearance of the
leading interviewee, which happens with probability $(1-P(X_1))$.

\item Next interviewee should be a new non-leading interviewee. As
there are $k$ non-leading interviewees among the $k+1$ distinct
interviewees interviewed so far. Therefore the probability is
$\left(\frac{k}{k+1}\right)$
\end{enumerate}

The total probability is the product of the probabilities of the above
two events together, i.e. 
\[P(X_3) = \left(1-\frac{1}{(1+p)(n-k)+1}\right)\frac{k}{k+1}\]

\item And according to the proposed strategy, we will always reject
this interviewee. So the success probability remains the same, but the
number of distinct interviewees interviewed has been increased by one.
Thus, the probability is. 
\[P(S_1|X_3) = \Phi_n(k+1)\]
\end{itemize} 
\end{description}
By the rule of total probability,
\[ P(S_1)=P(S_1|X_1)P(X_1) + P(S_1|X_2)P(X_2) + P(S_1|X_3)P(X_3) \]
\begin{equation*}
\begin{split}
P(S_1) &=
\left(\frac{k}{n}a\right)+\left(\left(1-p\right)\left(\frac{k+1}{n}\right)
+ p\left(\Phi_n(k+1)\right)\right)\left((1-a)\frac{1}{k+1}\right) \\
&+ \Phi_n(k+1)\left((1-a)\frac{k}{k+1}\right)\\
&= \frac{ka+(1-a)(1-p)}{n} + \frac{(1-a)(p+k)}{k+1}\Phi_n(k+1)\\
\end{split}
\end{equation*}
where $a = \frac{1}{(1+p)(n-k)+1}$.
    
\item \textit{$\bar{E}:$ The leading interviewee will not appear
again}\\ As each interviewee may appear a second time with probability
$p$, the probability that the leading interviewee will not appear
again is 
\[P(\bar{E}) = 1-p\] 
Let $P(S_2)$ be $\Phi_n(k)|\Bar{E}$ i.e, probability obtained with
respect to event $\Bar{E}$.\\ 
For the next interviewee that arrives following scenarios are
possible:
\begin{description}
\item[Case-1:] \textit{$Y_1$: Next interviewee is a new leading
interviewee.}
\begin{itemize}
\item Since the current leading interviewee cannot appear again, the
next interviewee must necessarily be distinct. Among $k+1$ distinct
interviewees observed so far, only one can be the leading interviewee.
Therefore, the probability that the next interviewee is a new leading
interviewee is \[ P(Y_1) = \frac{1}{k+1}\]

\item And according to the proposed strategy, we will accept this
interviewee with probability $1-p$ and reject him/her with probability
$p$. Once rejected, the interview process will continue for the next
interviewee. Therefore, the probability $P(S_2)$ with respect to event
$Y_1$ is 
\[P(S_2|Y_1) = \left(1-p\right)\left(\frac{k+1}{n}\right) +
p\Phi_n(k+1)\]

This is similar to $P(S_1|X_2)$
\end{itemize}

\item[Case-2:] \textit{$Y_2$: Next interviewee is a new non-leading
interviewee.}
\begin{itemize}
\item Since the current leading interviewee cannot appear again, the
next interviewee must necessarily be distinct. Among $k+1$ distinct
interviewees observed so far, $k$ distinct interviewees will be
non-leading interviewees. Therefore, the probability that the next
interviewee is a new non-leading interviewee is 
\[ P(Y_2) = \frac{k}{k+1}\]

\item And according to the proposed strategy, we will always reject
this interviewee. 
As now the interviewer has interviewed $k+1$ distinct interviewees.
Therefore, the probability is 
\[P(S_2|Y_2) = \Phi_n(k+1)\]
\end{itemize}  
\end{description}
By applying the rule of total conditional probability,
\[ P(S_2) = P(S_2|Y_1)P(Y_1) + P(S_2|Y_2)P(Y_2) \]
\begin{equation*}
\begin{split}
P(S_2) &= \left(\left(1-p\right)\left(\frac{k+1}{n}\right) +
p\Phi_n(k+1)\right)\left(\frac{1}{k+1}\right) +
\Phi_n(k+1)\left(\frac{k}{k+1}\right) \\
&= \frac{1-p}{n} + \frac{p+k}{k+1}\Phi_n(k+1)
\end{split}
\end{equation*}
\end{enumerate}
By combining the probabilities from both event $E$ and $\Bar{E}$ using
the rule of total probability, 
\[ \Phi_n(k) = P(E)P(S_1) + P(\Bar{E})P(S_2) \]
\begin{equation*}
\begin{split}
\Phi_n(k) &= p\left(\frac{ka+(1-a)(1-p)}{n} +
\frac{(1-a)(p+k)}{k+1}\Phi_n(k+1)\right) \\ &+
(1-p)\left(\frac{1-p}{n} + \frac{p+k}{k+1}\Phi_n(k+1)\right) \\
&= \frac{pak + (1-p)(1-pa)}{n} + \frac{(p+k)(1-pa)}{k+1}\Phi_n(k+1)
\end{split}
\end{equation*}
where 
$a = \frac{1}{(1+p)(n-k)+1}$.
\end{proof}

\begin{prop}
\label{prop:probArrival.7} The sequence of functions $\hat{\Phi}_n(x)
:= \Phi_n(\floor{nx})$ converges uniformly on $[0,1]$ for $\Phi(1) =p$
to
\begin{equation*}
\begin{split}
\Phi^\prime(x) = \left( \frac{1-p}{x} +\frac{p}{(1+p)(1-x)} \right) 
\Phi(x) - \left( \frac{px}{(1+p)(1-x)} + (1-p) \right)
\end{split}
\end{equation*}
\end{prop}
\begin{proof}
From Proposition \ref{prop:probArrival.1} we know that,
\[
    \Phi_n(k) = G_n(k) + H_n(k)\Phi_n(k+1)
\]
where $G_n(k) = \frac{pak + (1-p)(1-pa)}{n}$ and $H_n(k) =
\frac{(p+k)(1-pa)}{k+1}$ and $a=\frac{1}{(1+p)(n-k)+1}$.\\ 
From Theorem \ref{thm:3}, where $h_n(x) := n(1-H_n(\floor{nx}))$ and
$g_n(x) := nG_n(\floor{nx})$ we have, 
\[ g_n(x) = n\left(\frac{pa\floor{nx} + (1-p)(1-pa)}{n} \right) =
pa\floor{nx} + (1-p)(1-pa) \] \[ h_n(x) =
n\left(1-\left(\frac{(p+\floor{nx})(1-pa)}{\floor{nx}+1}\right)\right)
\]
As $n\rightarrow \infty$, $g_n(x) \rightarrow g(x)$, $h_n(x)
\rightarrow h(x)$ and $\floor{nx} \approx nx$.
\begin{itemize}
\item Let us estimate the value of $g(x)$ where
$a=\frac{1}{(1+p)(n-nx)+1}$, 
\begin{equation*}
\begin{split}
g(x) &= \underset{n \rightarrow \infty}{\lim} \left( panx +
(1-p)(1-pa) \right) \\
&= \underset{n \rightarrow \infty}{\lim} \left(
panx \right) + (1-p) \underset{n \rightarrow \infty}{\lim} \left(1-pa
\right)
\end{split}
\end{equation*}
Now let us consider each part separately as 
$\underset{n \rightarrow \infty}{\lim}$.
\begin{equation*}
\begin{split}
panx &= \frac{pnx}{(1+p)(n-nx)+1} = \frac{pnx}{n(1+p)(1-x)+1} \\
&= \frac{pnx}{(1+p)n(1-x)} \left(\frac{1}{1+\frac{1}{(1+p)n(1-x)}}
\right)
\end{split}
\end{equation*}
Using expansion for $\frac{1}{1+\epsilon}$ where $\epsilon =
\frac{1}{(1+p)n(1-x)}$. We have,
\begin{equation} \label{eq:panx}
\begin{split}
panx &= \frac{px}{(1+p)(1-x)} \left(1-\frac{1}{(1+p)n(1-x)} +
O(n^{-2}) \right) \\ %
&= \frac{px}{(1+p)(1-x)} + O(n^{-1})
\end{split}
\end{equation}
Let's now estimate $1-pa$,
\begin{equation*}
\begin{split}
pa &= \frac{p}{(1+p)(n-nx)+1} = \frac{p}{n(1+p)(1-x)+1} \\ %
&= \frac{p}{(1+p)n(1-x)} \left(\frac{1}{1+\frac{1}{(1+p)n(1-x)}}
\right)
\end{split}
\end{equation*}
Using expansion for $\frac{1}{1+\epsilon}$ where $\epsilon =
\frac{1}{(1+p)n(1-x)}$. We have,
\begin{equation} \label{eq:1minuspa}
\begin{split}
pa &= \frac{p}{(1+p)n(1-x)} \left(1-\frac{1}{(1+p)n(1-x)} + O(n^{-2})
\right) \\ &= \frac{p}{(1+p)n(1-x)} + O(n^{-2}) \\ %
1-pa &= 1-\frac{p}{(1+p)n(1-x)} + O(n^{-2})
\end{split}
\end{equation}
Substituting values for $(panx)$ and $(1-pa)$ from equation
\ref{eq:panx} and \ref{eq:1minuspa}. We get, %
as $n\rightarrow \infty$
\begin{equation*}
g(x) = \frac{px}{(1+p)(1-x)} + O(n^{-1}) + (1-p) 
\left(1-\frac{p}{(1+p)n(1-x)} + O(n^{-2}) \right)
\end{equation*}

As $n \rightarrow \infty$ then $O(n^{-1})$ and higher order term
vanishes, So,
\begin{equation} \label{eq:gx}
g(x) = \frac{px}{(1+p)(1-x)} + (1-p)
\end{equation}
\item Let us estimate the value of $h(x)$:
\[ h(x) = \underset{n \rightarrow
\infty}{\lim}n\left(1-\left(\frac{p+nx}{nx+1} (1-pa)\right)\right) \]
Let us consider $\underset{n \rightarrow \infty}{\lim}
\left(\frac{p+nx}{nx+1} \right)$.
\begin{equation*}
\begin{split}
\underset{n \rightarrow \infty}{\lim} \left(\frac{p+nx}{nx+1} \right)
= \left(\frac{p+nx}{nx} \right) \left( \frac{1}{1+\frac{1}{nx}}
\right)
\end{split}
\end{equation*}
Using expansion for $\frac{1}{1+\epsilon}$ where $\epsilon =
\frac{1}{nx}$. We have,
\begin{equation}
\label{eq:pplusnx}
\begin{split}
\underset{n \rightarrow \infty}{\lim} \left(\frac{p+nx}{nx+1} \right)
&= \left(\frac{p+nx}{nx} \right) \left( 1- \frac{1}{nx} + O(n^{-2})
\right) \\ %
&= \left(\frac{p}{nx} +1 \right) \left( 1- \frac{1}{nx} + O(n^{-2})
\right)\\ &= \frac{p}{nx} +1 - \frac{1}{nx} + O(n^{-2}) \\ %
&= 1-\frac{1-p}{nx} + O(n^{-2})
\end{split}
\end{equation}
Substituting values from equations \ref{eq:1minuspa} and
\ref{eq:pplusnx}. We get,
\begin{equation*}
\begin{split}
h(x) &= n\left(1-\left(1-\frac{1-p}{nx} + O(n^{-2}) \right)\left(
1-\frac{p}{(1+p)n(1-x)} + O(n^{-2})\right)\right) \\ %
&= n\left(1-\left(1-\frac{1-p}{nx} -\frac{p}{(1+p)n(1-x)} + O(n^{-2})
\right)\right) \\  %
&= n\left(\frac{1-p}{nx} +\frac{p}{(1+p)n(1-x)} + O(n^{-2})\right) \\
&= \left(\frac{1-p}{x} +\frac{p}{(1+p)(1-x)} + O(n^{-1})\right)
\end{split}
\end{equation*}
As $n \rightarrow \infty$ then $O(n^{-1})$ and higher order term
vanishes, So,
\begin{equation} \label{eq:hx}
h(x) = \frac{1-p}{x} +\frac{p}{(1+p)(1-x)}
\end{equation}
\end{itemize}
Using Theorem \ref{thm:3} and Substituting $g(x)$ and $h(x)$ from
above equations \ref{eq:gx}, \ref{eq:hx}, we get,
\begin{equation*}
\Phi^\prime(x) = \left( \frac{1-p}{x} +\frac{p}{(1+p)(1-x)} \right)  \Phi(x) - \left( \frac{px}{(1+p)(1-x)} + (1-p) \right)
\end{equation*}
\end{proof}

\subsubsection{Calculation of $\Psi_n(k)$}

We calculate the value of $\Psi_n(k)$ as follows: 
\begin{prop}
\label{prop:probArrival.3} For all natural numbers $n$ and $k$ where
$k < n$, we have, \[ \Psi_n(k) =
\frac{1-p}{n} + \frac{p}{k+1}\Phi_n(k+1) + \frac{k}{k+1}\Psi_n(k+1) \]
and $\Phi_n(n) = p, \Psi_n(n) = 0$
\end{prop}
\begin{proof}
By hypothesis, when all $n$ distinct interviewees have been
interviewed, then the leading interviewee is the overall best, and
he/she has already appeared twice, so he/she will not appear again.
Thus, $\Psi_n(n)=0$,\\ %
Now for the general case $k$, since the leading interviewee has
already appeared twice, he/she will not appear again. So the next
interviewee interviewed can only be one of the following:
\begin{description}
\item[Case-1:] \textit{$Z_1$: Next interviewee is a new leading
interviewee.}
\begin{itemize}
\item Since the current leading interviewee cannot appear again, the
next interviewee must necessarily be distinct. Among $k+1$ distinct
interviewees observed so far, only one can be the leading interviewee.
Therefore, the probability that the next interviewee is a new leading
interviewee is %
\[ P(Z_1) = \frac{1}{k+1}\]
\item And according to the proposed strategy, we will accept this
interviewee with probability $1-p$ and reject him/her with probability
$p$. Therefore, the probability $\Psi_n(k)$ with respect to event
$Z_1$ is %
\[\Psi_n(k)|Z_1 = \left(1-p\right)\left(\frac{k+1}{n}\right) +
p\left(\Phi_n(k+1)\right)\]
\end{itemize}

\item[Case-2:] \textit{$Z_2$: Next interviewee is a new non-leading
interviewee.}
\begin{itemize}
\item Since the current leading interviewee cannot appear again, the
next interviewee must necessarily be distinct. Among $k+1$ distinct
interviewees observed so far, $k$ distinct interviewees will be
non-leading interviewees. Therefore, the probability that the next
interviewee is a new non-leading interviewee is %
\[ P(Z_2) = \frac{k}{k+1}\]

\item And according to the proposed strategy, we will always reject
this interviewee. 
Therefore, the probability is %
\[\Psi_n(k)|Z_2 = \Psi_n(k+1)\]
\end{itemize}   
\end{description}
By the rule of total probability,
\[\Psi_n(k) = (\Psi_n(k)|Z_1)P(Z_1) + (\Psi_n(k)|Z_2)P(Z_2) \]
\begin{equation*}
\begin{split}
\Psi_n(k) &= \left((1-p)\left(\frac{k+1}{n}\right) +
p\Phi_n(k+1)\right)\frac{1}{k+1} + \Psi_n(k+1)\frac{k}{k+1} \\%
&= (1-p)\frac{k+1}{n}\frac{1}{k+1} + p\Phi_n(k+1)\frac{1}{k+1} +
\Psi_n(k+1)\frac{k}{k+1} \\ %
&= \frac{1-p}{n} + \frac{p}{k+1}\Phi_n(k+1) + \frac{k}{k+1}\Psi_n(k+1)
\end{split}
\end{equation*}
\end{proof}

\begin{prop}
\label{prop:probArrival.8} 
The sequence of functions $\hat{\Psi}_n(x) := \Psi_n(\floor{nx})$
converges uniformly on $[0,1]$ for $\Psi(1)=0$ to
\begin{equation*}
\Psi^\prime(x) = \frac{1}{x} \Psi(x) - \left( 1-p +\frac{p}{x}\Phi(x)
\right)
\end{equation*}
\end{prop}
\begin{proof}
As %
\[ \Psi_n(k) = G_n(k) + H_n(k)\Psi_n(k+1) \] where $G_n(k) =
\frac{1-p}{n} + \frac{p}{k+1}\Phi_n(k+1)$ and $H_n(k) =
\frac{k}{k+1}$.\\ %
To use 
Theorem \ref{thm:3}  we let 
$h_n(x) := n(1-H_n(\floor{nx}))$ and $g_n(x) := nG_n(\floor{nx})$, now
\[ g_n(x) = n\left(\frac{1-p}{n} +
\frac{p}{\floor{nx}+1}\Phi_n(\floor{nx}+1) \right) = 1-p +
\frac{pn}{\floor{nx}+1}\Phi_n(\floor{nx}+1) \] 
\[ h_n(x) =
n\left(1-\left(\frac{\floor{nx}}{\floor{nx}+1}\right)\right) =
\left(\frac{n}{\floor{nx}+1}\right) \] 
As $n\rightarrow \infty$, $g_n(x) \rightarrow g(x)$, $h_n(x)
\rightarrow h(x)$ and $\floor{nx} \approx nx$.
\begin{itemize}
\item Let us estimate the value of $g(x)$, %
as $n\rightarrow \infty$ :
\[ g(x) = 1-p + \underset{n \rightarrow \infty}{\lim} \left(
\frac{pn}{nx+1}\Phi_n(nx+1) \right) \]
Now let us consider $\underset{n \rightarrow \infty}{\lim}
\frac{pn}{nx+1}$.
\begin{equation*}
\begin{split}
\frac{pn}{nx+1} = \frac{pn}{nx} \left(\frac{1}{1+\frac{1}{nx}} \right)
\end{split}
\end{equation*}
Using expansion for $\frac{1}{1+\epsilon}$ where $\epsilon =
\frac{1}{nx}$. So,
\begin{equation} \label{eq:pn}
\begin{split}
\frac{pn}{nx+1} &= \frac{p}{x} \left(1-\frac{1}{nx} + O(n^{-2})
\right) \\ %
&= \frac{p}{x} + O(n^{-1})
\end{split}   
\end{equation}
From Proposition~\ref{prop:probArrival.1}, and as $n\rightarrow \infty$,
$\Phi_n(nx+1)=\Phi_n(nx)$
\begin{equation} \label{eq:phinx}
\underset{n \rightarrow \infty}{\lim} \Phi_n(nx+1)= 
\Phi(x)
\end{equation}
Substituting values for $(\frac{pn}{nx+1})$ and $\Phi_n(nx+1)$ from
equations \ref{eq:pn}, \ref{eq:phinx}, We get,
\begin{equation*}
g(x) = 1-p + \left(\frac{p}{x} + O(n^{-1}) \right) \Phi(x)
\end{equation*}
As $n \rightarrow \infty$ then $O(n^{-1})$ and higher order term
vanishes, So,
\begin{equation} \label{eq:gx2}
g(x) = 1-p + \frac{p}{x}\Phi(x)
\end{equation}

\item Let us estimate the value of $h(x)$: 
\[ h(x) = \underset{n \rightarrow \infty}{\lim} \frac{n}{nx+1} =
\underset{n \rightarrow \infty}{\lim} \frac{n}{nx} \left(
\frac{1}{1+\frac{1}{nx}} \right) \] 
Using expansion for $\frac{1}{1+\epsilon}$ where $\epsilon =
\frac{1}{nx}$. So,
\begin{equation*}
\begin{split}
h(x) &= \underset{n \rightarrow \infty}{\lim} \frac{1}{x} \left( 1-
\frac{1}{nx} + O(n^{-2}) \right) \\ &= \frac{1}{x} + O(n^{-1})
\end{split}
\end{equation*}
As $n \rightarrow \infty$ then $O(n^{-1})$ and higher order term
vanishes, So,
\begin{equation} \label{eq:hx2}
h(x) = \frac{1}{x}
\end{equation}
\end{itemize}
Using Theorem \ref{thm:3} and Substituting $g(x)$ and $h(x)$ from
above equations \ref{eq:gx2}, \ref{eq:hx2}, we get,
\begin{equation*}
\Psi^\prime(x) = \frac{1}{x} \Psi(x) - \left( 1-p +\frac{p}{x}\Phi(x)
\right)
\end{equation*}
\end{proof}

\subsubsection{Calculation of $\Upsilon_n(k)$}

We calculate the value of $\Upsilon_n(k)$ as follows:
\begin{prop}
\label{prop:probArrival.5} For all natural numbers $n$ and $k$ where
$1< k \leq n$, we have, 
\[ \Upsilon_n(k) = \frac{1}{k} +
\left(1-\frac{p}{(1+p)(n-k+1)+1}\right)\left(1-\frac{1}{k}\right)\Upsilon_n(k-1)
\]
and $\Upsilon_n(1) = 1$
\end{prop}
\begin{proof}
By hypothesis, when the first (distinct) interviewee makes his/her
initial appearance, that interviewee is necessarily the leading
interviewee. Therefore, the probability that the leading interviewee
appears only once when the first distinct interviewee is interviewed,
i.e., $\Upsilon_n(1)$, is equal to $1$. \\
Now, for the general case of $k$ distinct interviews, let's consider
two mutually exclusive events.
\begin{enumerate}
\item \textit{$E$: The leading interviewee has appeared only once
when the $(k-1)^{th}$ distinct interviewee makes his/her initial
appearance.}\\ 
Let $P(Q_1)$ be $\Upsilon_n(k)|E$ i.e. probability conditioned on
event $E$.\\ 
This case can be categorized further into two mutually exclusive
events:
\begin{description}
\item[Case-1:] \textit{$R$: When the future appearance of the leading
interviewee observed in $(k-1)$ distinct interviewee interviews is
possible}\\
As each interviewee may appear a second time with probability $p$.
Therefore,
\[P(R) = p\] 
Now, for the leading interviewee to be observed only once in $k$
distinct interviewee interviews, one of the following two events
should happen.
\begin{itemize}
\item \textit{$T_1$: The leading interviewee appears again before the
$k^{th}$ distinct interviewee is interviewed and the $k^{th}$ distinct
interviewee is a new leading interviewee.}\\ 
Since only one remaining appearance of the current leading interviewee
is among the unobserved interviewees, the probability that he/she
arrives again after $(k-1)^{th}$ distinct interviewee has been
interviewed is $\frac{1}{(1+p)(n-(k-1))+1}$. Furthermore, the
probability that the $k^{th}$ distinct interviewee is a new leading
interviewee is $\frac{1}{k}$. Hence, 
\[P(T_1) = \frac{1}{k}\frac{1}{(1+p)(n-(k-1))+1}\]

\item \textit{$T_2$: The leading interviewee does not appear again
before $k^{th}$ distinct interviewee is interviewed.}\\ 
The probability that the leading interviewee does not appear again
before the $k^{th}$ distinct interviewee is one minus the probability
of the second appearance of the leading interviewee. Therefore,
\[P(T_2) = 1-\frac{1}{(1+p)(n-(k-1))+1}\].
\end{itemize}
By using the rule of total probability, 
\[ P(Q_1|R) = P(T_1)+P(T_2) \]
\begin{equation*}
\begin{split}
P(Q_1|R) &=
\frac{1}{k}\frac{1}{(1+p)(n-(k-1))+1}+1-\frac{1}{(1+p)(n-(k-1))+1} \\
&= 1- \left(1-\frac{1}{k} \right)\left(\frac{1}{(1+p)(n-k+1)+1}\right)
\end{split}
\end{equation*}
    
\item[Case-2:] \textit{$\Bar{R}$: When the future appearance is not
possible for the leading interviewee observed in $(k-1)$ distinct
interviewee interviews.}\\ 
As each interviewee may appear a second time with probability $p$.
Therefore, 
\[P(\Bar{R}) =1-p\] 
Since the future appearance is not possible for the leading
interviewee, the next interviewee can only be a new distinct
interviewee, so in either case, when this new distinct interviewee is
a new leading or not, we will always have the case when the leading
interviewee is observed once in $k$ distinct interviews. Therefore, 
\[ P(Q_1|\Bar{R}) = 1 \]
\end{description}
By the rule of total probability,
\[ P(Q_1)=P(Q_1|R)P(R) + P(Q_1|\Bar{R})P(\Bar{R}) \]
\begin{equation*}
\begin{split}
P(Q_1)&=\left(1- \left(1-\frac{1}{k}
\right)\left(\frac{1}{(1+p)(n-k+1)+1}\right)\right)p + (1-p) \\ %
&= p-p\left(1-\frac{1}{k} \right)\left(\frac{1}{(1+p)(n-k+1)+1}\right)
+ 1-p \\ &= 1-p\left(1-\frac{1}{k}
\right)\left(\frac{1}{(1+p)(n-k+1)+1}\right)
\end{split}
\end{equation*}

\item \textit{$\Bar{E}$: The leading interviewee has appeared twice
when the $(k-1)^{th}$ distinct interviewee makes his/her initial
appearance.}\\ 
Let $P(Q_2)$ be $\Upsilon_n(k)|\Bar{E}$ i.e. probability conditioned
on $\Bar{E}$.\\ 
As the current leading interviewee has already appeared twice, for the
single appearance of the leading interview to be present in $k$
distinct interviewees, the $k^{th}$ distinct interviewee must be a new
leading interviewee. Therefore,
\[ P(Q_2) = \frac{1}{k} \]
\end{enumerate}
As the probability of event $E$ is $\Upsilon_n(k-1)$ and event
$\Bar{E}$ is $1-\Upsilon_n(k)$. Therefore, by combining the
probabilities from both events $E$ and $\Bar{E}$ using the rule of
total probability, we get, 
\[ \Upsilon_n(k) = P(E)P(Q_1) + P(\Bar{E})P(Q_2) \]
\begin{equation*}
\begin{split}
\Upsilon_n(k) &= \Upsilon_n(k-1)\left(1-p\left(1-\frac{1}{k}
\right)\left(\frac{1}{(1+p)(n-k+1)+1}\right)\right)
+(1-\Upsilon_n(k-1))\left(\frac{1}{k}\right) \\
&= \frac{1}{k}+
\left(1-\frac{1}{k}\right)\left(1-\frac{p}{(1+p)(n-k+1)+1}\right)\Upsilon_n(k-1)
\end{split}
\end{equation*}
\end{proof}

\begin{prop}
\label{prop:probArrival.9} The sequence of functions
$\hat{\Upsilon}_n(x) := \Upsilon_n(\floor{nx})$ converges uniformly on
$[0,1]$ for $\Upsilon(0)=1$ to
\begin{equation*}
\Upsilon^\prime(x) = -\left(\frac{1}{x} + \frac{p}{(1+p)(1-x)} \right)
\Upsilon(x) + \frac{1}{x}
\end{equation*}
\end{prop}
\begin{proof}
From Proposition \ref{prop:probArrival.5} we know that,
\[
    \Upsilon_n(k) = G_n(k) + H_n(k)\Upsilon_n(k-1)
\]
where $G_n(k) = \frac{1}{k}$ and $H_n(k) =
\left(1-\frac{p}{(1+p)(n-k+1)+1}\right) 
\left(1-\frac{1}{k}\right)$.\\ 
We use 
Theorem \ref{thm:2} with 
$h_n(x) := n(1-H_n(\floor{nx}))$ and $g_n(x) := nG_n(\floor{nx})$. Now
\[ g_n(x) = n\frac{1}{\floor{nx}} \]
\begin{equation*}
\begin{split}
h_n(x) &= n\left(1- \left(1-\frac{p}{(1+p)(n-\floor{nx}+1)+1}\right) 
\left(1-\frac{1}{\floor{nx}}\right)\right) \\ %
&= n\left(\frac{1}{\floor{nx}} +
\left(\frac{p}{(1+p)(n-\floor{nx}+1)+1} \right)
\left(1-\frac{1}{\floor{nx}}\right)\right)
\end{split}
\end{equation*}
As $n\rightarrow \infty$, $g_n(x) \rightarrow g(x)$, $h_n(x)
\rightarrow h(x)$ and $\floor{nx} \approx nx$.
\begin{itemize}
\item Let us estimate the value of $g(x)$:
\begin{equation} \label{eq:gx3}
g(x) = \underset{n \rightarrow \infty}{\lim} \left( \frac{1}{x}
\right) = \frac{1}{x}
\end{equation}
\item Let us estimate the value of $h(x)$:
\begin{equation*}
\begin{split}
h(x) &= \underset{n \rightarrow \infty}{\lim} n\left(\frac{1}{nx} +
\left(\frac{p}{(1+p)(n-nx+1)+1} \right) 
\left(1-\frac{1}{nx}\right)\right) \\ 
&= \underset{n \rightarrow \infty}{\lim} \left(\frac{n}{nx} + n
\left(1-\frac{1}{nx}\right) \left(\frac{p}{(1+p)(n-nx+1)+1}
\right)\right) \\ 
&= \underset{n \rightarrow \infty}{\lim} \left(\frac{1}{x} + n
\left(1-\frac{1}{nx}\right) \left(\frac{p}{(1+p)(n-nx)+(p+1)+1}
\right)\right) \\ 
&= \frac{1}{x} +\underset{n \rightarrow \infty}{\lim} \left( n
\left(1-\frac{1}{nx}\right) \left(\frac{p}{(1+p)n(1-x)}\right) 
\left(\frac{1}{1+\frac{2+p}{(1+p)n(1-x)}} \right)\right) 
\end{split}
\end{equation*} 
Using expansion for $\frac{1}{1+\epsilon}$ where $\epsilon =
\frac{2+p}{(1+p)n(1-x)}$. So,
\begin{equation*}
\begin{split}
h(x) &= \frac{1}{x} +\underset{n \rightarrow \infty}{\lim}
\left(\left(1-\frac{1}{nx}\right) \frac{p}{(1+p)(1-x)} 
\left(1-\frac{2+p}{(1+p)n(1-x)} + O(n^{-2}) \right)\right)\\ %
&=\frac{1}{x} +\underset{n \rightarrow \infty}{\lim}
\left(\left(1-\frac{1}{nx}\right) \left(\frac{p}{(1+p)(1-x)} +
O(n^{-1}) \right)\right)\\ %
&=\frac{1}{x} +\underset{n \rightarrow \infty}{\lim}
\left(\frac{p}{(1+p)(1-x)} + O(n^{-1}) \right)
\end{split}
\end{equation*}
As $n \rightarrow \infty$ then $O(n^{-1})$ and higher order term
vanishes, so,
\begin{equation} \label{eq:hx3}
h(x) = \frac{1}{x} +\frac{p}{(1+p)(1-x)}
\end{equation}
\end{itemize}
Using Theorem \ref{thm:2} and substituting $g(x)$ and $h(x)$ from
above equations \ref{eq:gx3} and \ref{eq:hx3}, we get,
\begin{equation*}
\Upsilon^\prime(x) = -\left(\frac{1}{x} + \frac{p}{(1+p)(1-x)} \right)
\Upsilon(x) + \frac{1}{x}
\end{equation*}
\end{proof}

\subsubsection{Calculation of $F_n(k)$}

Overall success probability $F_n(k)$ can be obtained by combining the
probabilities of events $\chi_1$ and $\chi_2$ by using the total
probability rule.\\
From Proposition \ref{prop:probArrival.1} we know $\Phi_n(k)$, the
probability of event $\chi_1$.\\
From Proposition \ref{prop:probArrival.3}, we know $\Psi_n(k)$, the
probability of event $\chi_2$.\\
From Proposition \ref{prop:probArrival.5}, we know the value of
$\Upsilon_n(k)$.\\
Now $\Bar{\Upsilon}_n(k)$, the probability that the leading
interviewee has appeared twice when the $(k-1)^{th}$ distinct
interviewee makes his/her initial appearance is one minus
$\Upsilon_n(k)$. Therefore,
\[\Bar{\Upsilon}_n(k) = 1- \Upsilon_n(k)\] 
By combining the above probabilities using the law of conditional
probability, we can obtain,
\begin{equation} \label{eq:probArrivalpsuccess}
\begin{split}
F_n(k) &= \Upsilon_n(k)\Phi_n(k) + \Bar{\Upsilon}_n(k)\Psi_n(k) \\
    &= \Upsilon_n(k)\Phi_n(k) + (1-\Upsilon_n(k))\Psi_n(k)
\end{split}
\end{equation}
From Proposition \ref{prop:probArrival.1} we can compute $\phi_n(k)$
from $\phi_n(k+1)$,
and from Proposition \ref{prop:probArrival.3} we can compute $\psi_n(k)$ 
from $\psi_n(k+1)$ and finally 
from Proposition \ref{prop:probArrival.5} we can compute
$\Upsilon_n(k)$ from $\Upsilon_n(k-1)$ in $O(1)$ computation time.
Thus, values for all $1\leq k \leq n$ can be computed 
in $O(n)$ time.
From above results 
using Algorithm \ref{alg:getmat}, \ref{alg:getPrTh}, we can obtain the
success probability($F(k_n)$) in $O(n)$ computation time.

\begin{prop}
\label{prop:probArrival.10} 
The sequence of functions $\hat{F}_n(x) := F_n(\floor{nx})$ converges
uniformly on $[0,1]$ to
\begin{equation*}
F(x) = \Upsilon(x)\Phi(x) + (1-\Upsilon(x))\Psi(x)
\end{equation*}
\end{prop}
\begin{proof}
From equation \ref{eq:probArrivalpsuccess}, we know that,
\begin{gather*}
F_n(k) = \Upsilon_n(k)\Phi_n(k) + (1-\Upsilon_n(k))\Psi_n(k)
\end{gather*}
Substituting $k = \floor{nx}$,
\begin{equation*}
F_n(k) = F_n(\floor{nx}) = \Upsilon_n(\floor{nx})\Phi_n(\floor{nx}) +
(1-\Upsilon_n(\floor{nx}))\Psi_n(\floor{nx})
\end{equation*}
using the notations introduced above,
\begin{equation*}
\hat{F}_n(x) = \hat{\Upsilon}_n(x)\hat{\Phi}_n(x) +
(1-\hat{\Upsilon}_n(x))\hat{\Psi}_n(x)
\end{equation*}
For uniform convergence at $n \rightarrow \infty$,
\begin{equation*}
\underset{n \rightarrow \infty}{\lim} \hat{F}_n(x) = \underset{n
\rightarrow \infty}{\lim}\hat{\Upsilon}_n(x)\hat{\Phi}_n(x) +
\underset{n \rightarrow
\infty}{\lim}(1-\hat{\Upsilon}_n(x))\hat{\Psi}_n(x)
\end{equation*}
As per Theorem \ref{thm:1}, at $n \rightarrow \infty, \ \ F(x) =
\hat{F}_n(x), \Upsilon(x) =\hat{\Upsilon}_n(x), \Phi(x) =
\hat{\Phi}_n(x),$\\ $\Psi(x) = \hat{\Psi}_n(x)$, so,
\begin{equation*}
F(x) = \Upsilon(x)\Phi(x) + (1-\Upsilon(x))\Psi(x)
\end{equation*}
\end{proof}

\begin{algorithm}[h!]
\DontPrintSemicolon
  \KwInput{$n, p$}
  \KwOutput{$\Phi, \Psi, \Upsilon$}
  \SetKwFunction{FgetPW}{get\_PW\_arrays}
  \SetKwProg{Fn}{Def}{:}{}
  \Fn{\FgetPW{}}{
        $\Phi[n] \longleftarrow p$\\
        $\Psi[n] \longleftarrow 0$\\
         \For{$k \gets n-1$ down \KwTo $0$} 
          {
	  $a = 1/((1+p)*(n-k)+1)$\\
            $\Phi[k] = (p*a*k+(1-p)*(1-p*a))/n + (p+k)*(1-p*a)*\Phi[k+1]/(k+1)$\\
            $\Psi[k] = (1-p)/n + (p*\Phi[k+1] + k*\Psi[k+1])/(k+1)$\\
            
          }
        $\Upsilon[1] \longleftarrow 1$\\
        \For{$k=2$ \KwTo $n$} 
          {
          $\Upsilon[k] = 1/k + (1-p/((1+p)(n-k+1)+1))*(1-1/k)*  \Upsilon[k-1]$\\
            
          }
        \KwRet $\Phi, \Psi, \Upsilon$
     }
\caption{Get $\Phi, \Psi, \Upsilon$ arrays}
\label{alg:getmat}
\end{algorithm}

\begin{algorithm}[h!]
\DontPrintSemicolon
  
  \KwInput{$n, \Phi, \Psi, \Upsilon$}
  \KwOutput{$F[k_n],k_n$}
  \SetKwFunction{FgetThp}{get\_optimal\_probability}

  \SetKwProg{Fn}{Def}{:}{}
  \Fn{\FgetThp{}}{
    \FgetPW{} \tcp*{function from Algorithm-\ref{alg:getmat}}
    $F[1,\dots,n] \longleftarrow 0$\\
    $F(k_n) \longleftarrow 0$\\
    $k_n \longleftarrow 0$\\
   \For{$k=1$ \KwTo $n$}
   {
        $F[k] = \Phi[k]*\Upsilon[k] + (1-\Upsilon[k])*\Psi[k]$\\
        \If{$F[k] > F(k_n)$}{
            $F(k_n) \longleftarrow F[k]$\\
            $k_n \longleftarrow k$
        }
    }
    \KwRet $F(k_n),k_n$
  }
\caption{Get Optimal Probability ($F(k_n)$)}
\label{alg:getPrTh}
\end{algorithm}

\subsection{Special cases for probabilistic second arrival}

Here we will be showing that the model we have derived in the earlier
section reduces to the standard models if 
$p=0$  or 
$p=1$.

\subsubsection{Case $p=0$ (Classical Secretary Problem)}
When $p=0$, no interviewee arrives a second time, and each interviewee
is observed exactly once. 
On substituting $p=0$ in the equation from Proposition
\ref{prop:probArrival.1}, we get,
\begin{gather*}
  \Phi_n(k) = \frac{1}{n} + \frac{k}{k+1} \Phi_n(k+1)
\end{gather*}
Since each interviewee only arrives once, $\Upsilon_n(k)$ will always
be $1$. Therefore,
\begin{equation*}
\begin{split}
F_n(k) &= \Upsilon_n(k)\Phi_n(k) + \Bar{\Upsilon}_n(k)\Psi_n(k) \\ %
&= \Phi_n(k) + (1-1)\Psi_n(k) \\ &= \Phi_n(k)
\end{split}
\end{equation*}
Therefore,
\begin{gather*}
F_n(k) = \frac{1}{n} + \frac{k}{k+1} F_n(k+1) \ \ for \ \ 1 \leq k
\leq n
\end{gather*}
This is the same recursive function obtained by Ribas\cite{Ribas2018}
for the Classical Secretary Problem, which gives Optimal Probability
as $1/e$ at optimal threshold $n/e$, see Ribas\cite{Ribas2018} for
details.

\subsubsection{Case $p=1$ (Guaranteed Reappearance)}

When $p=1$, every interviewee arrives exactly twice. 

\begin{itemize}
\item \textit{Value for $\Phi_n(k)$ at $p=1$:}\\
On substituting $p=1$ in the equation from Proposition
\ref{prop:probArrival.1}, we get,
\begin{gather*}
\Phi_n(k) = \frac{ak}{n} + \frac{(1+k)(1-a)}{k+1} \Phi_n(k+1)
\end{gather*}
where $a = \frac{1}{2n-2k+1}$. Therefore,
\begin{equation*}
\begin{split}
\Phi_n(k) &= \frac{k}{n} \frac{1}{2n-2k+1} + \left(1-\frac{1}{2n-2k+1}
\right)\Phi_n(k+1) \\ %
&= \frac{k}{n} \frac{1}{2n-2k+1} + \frac{2n-2k}{2n-2k+1} \Phi_n(k+1)
\end{split}
\end{equation*}

\item 
On substituting $p=1$ in the equation from Proposition
\ref{prop:probArrival.3}, we get,
\begin{gather*}
\Psi_n(k) = \frac{1}{k+1} \Phi_n(k+1) + \frac{k}{k+1} \Psi_n(k+1)
\end{gather*}
    
\item 
On substituting $p=1$ in the equation from Proposition
\ref{prop:probArrival.5}, we get,
\begin{gather*}
\Upsilon_n(k) = \frac{1}{k} +
\left(1-\frac{1}{2n-2k+3}\right)\left(1-\frac{1}{k}\right) 
\Upsilon_n(k-1)
\end{gather*}
Ribas\cite{Ribas2018} has obtained $\Bar{\Upsilon}_n(k)$, the
probability that the leading interviewee has appeared twice when the
$(k-1)^{th}$ distinct interviewee makes his/her initial appearance.
As 
$\Bar{\Upsilon}_n(k) = 1-\Upsilon_n(k)$, 

\begin{equation*}
\begin{split}
1-\Bar{\Upsilon}_n(k) &= \frac{1}{k} +
\left(1-\frac{1}{2n-2k+3}\right) \left(1-\frac{1}{k}\right) 
(1-\Bar{\Upsilon}_n(k-1))\\ %
\Bar{\Upsilon}_n(k) &=1-\frac{1}{k} -
\left(1-\frac{1}{2n-2k+3}\right) \left(1-\frac{1}{k}\right)(1-\Bar{\Upsilon}_n(k-1))\\ %
&= \left(\frac{k-1}{k(2n-2k+3)}\right) +
\left(\frac{2(k-1)(n-k+1)}{k(2n-2k+3)}\right)\Bar{\Upsilon}_n(k-1)\\ 
\end{split}
\end{equation*}   
\end{itemize}

These equations (for case $p=1$) are the same recursive equation
obtained by Ribas\cite{Ribas2018}. So we can say that the optimal
probability of success is the same as the one obtained by
Ribas\cite{Ribas2018}. Therefore, the optimal probability is $0.76$ at
an optimal threshold of $0.47$.

\subsection{Numerical Simulation}

Table \ref{tab:getProptth} shows the optimal threshold $k_n$ together
with the corresponding optimal success probability $P(k_n)$ obtained
using Algorithm \ref{alg:getPrTh} for various values of $p$ when
$n=100$. The table also reports the normalized threshold $k_n/n$. The
results demonstrate how allowing reappearances (increasing $p$)
affects both the optimal stopping rule and the probability of
successfully selecting the best candidate.
\begin{table}[h!]
\centering
\renewcommand{\arraystretch}{1.4} 
\setlength{\tabcolsep}{10pt}      
\begin{tabular}{ |c||c|c|c|c|c|c|c|c|c| } 
\hline
$\mathbf{p}$ & 0 & 0.001 & 0.1 & 0.25 & 0.5 & 0.75 & 0.9 & 0.999 & 1 \\ 
\hline
$\mathbf{k_n}$ & 37 & 37 & 47 & 55 & 57 & 54 & 51 & 48 & 48 \\ 
\hline
$\mathbf{\frac{k_n}{n}}$ & 0.37 & 0.37 & 0.47 & 0.55 & 0.57 & 0.54 & 0.51 & 0.48 & 0.48 \\
\hline
$\mathbf{P(k_n)}$ & 0.371 & 0.372 & 0.484 & 0.597 & 0.6874 & 0.7328 & 0.7546 & 0.7695 & 0.7697 \\
\hline
\end{tabular}
\caption{Optimal values $k_n$, $k_n/n$, and $P(k_n)$ for different
values of $p$ when $n=100$} \label{tab:getProptth}
\end{table}

This figure illustrates the success probability $P(k)$ as a function
of the stopping threshold $k$ for $n=100$ under different values of
the reappearance probability $p$. Each curve corresponds to a
different $p$, and the dashed vertical lines indicate the respective
optimal thresholds. The results show that as $p$ increases, both the
optimal threshold and the maximum success probability increase,
demonstrating the positive effect of candidate reappearance on
selection success.

\begin{figure}[h!]
    \centering
    \includegraphics[width = 15cm]{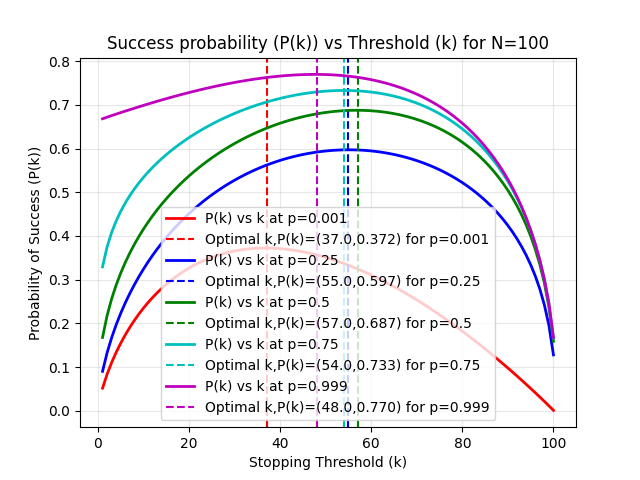}
    \caption{Success probability graph w.r.t. $k$ at $n=100$}
    \label{fig:probArrivalpsuccess}
\end{figure}

\section{Secretary problem with top-3 success}
\label{chap:relatedchap3}

In this case, each interview arrives exactly once.  %
But, we will be satisfied if the chosen candidate is amongst the best
(top) three.

We 
use the optimal strategy of the classical secretary problem, i.e.
interview $k$ interviewees, reject all of them but note down the
leading one among them.

After interviewing $k$ interviewees, select the first interviewee that
is better than all the previously seen interviewees.

\subsection{Success probability for top-3 secretary problem}

Let,\\ $P_n(k)$: Probability that the chosen interviewee is one of
top-3 assuming we have rejected the first $k$ interviewees.

\subsubsection{Calculation of $P_n(k)$}

For $P_n(k)$, we have the below proposition,
\begin{prop}
\label{prop:top3.1}
For all natural numbers $n$ and $k$ where $k<n$, we have,
\[P_n(k) = \frac{1}{k+1} \left(1-
\frac{\binom{n-3}{k+1}}{\binom{n}{k+1}}\right) + \frac{k}{k+1} 
P_n(k+1) \] and $P_n(n)=0$.
\end{prop}
\begin{proof}
By hypothesis, when an employer has already interviewed all $n$
interviewees, then the leading interviewee is the overall best, and
he/she has already appeared and was 
rejected, so the probability of selecting the best 
is $0$, or 
$P_n(n)=0$.\\
Now for the general case of $k$ interviews, let's look at time instant
$(k+1)$:\\
At $(k+1)^{th}$ time instant, either 
of the following events occurs:
\begin{itemize}
\item \textit{$E$: The $(k+1)^{th}$ interviewee is a new leading
interviewee.}\\ 
According to the proposed strategy, 
this interviewee will be accepted
immediately. Therefore, the success probability with respect to event
$E$ is the probability of choosing one of the top $3$. Therefore, 
\[ P(\mbox{success|}E) = P(\mbox{choosen in top }3|E) \] 
The complement of this event is the event that none of the first $k+1$
interviewees is among top three. The probability of this is:
\[P(\mbox{not in top } 3)=\frac{\binom{n-3}{k+1}}{\binom{n}{k+1}} \]
Here $\binom{n}{r}$ is the number of ways to choose $r$
items from $n$ items. 
\begin{equation} \label{eq:min3}
P(\mbox{in top } 3 | E) = 1- \frac{\binom{n-3}{k+1}}{\binom{n}{k+1}} 
\end{equation}
Now, the probability that the leading interviewee appears at $k+1$ is
\begin{equation} \label{eq:pe}
P(E)=\frac{1}{k+1}
\end{equation}
\item \textit{$\Bar{E}$: The $(k+1)^{th}$ interviewee is a
non-leading interviewee.}\\ 
According to the proposed strategy, this interviewee will not be
chosen (we will always reject him/her and will continue the process).
Therefore,
\begin{equation} \label{eq:patebar}
P_n(k)|\Bar{E} = P_n(k+1)
\end{equation}
And the probability that a non-leading interviewee appears at
$(k+1)^{th}$ interview is one minus the probability that a leading
interviewee does not appear at $(k+1)^{th}$ interview. Therefore,
\begin{equation} \label{eq:pebar}
P(\Bar{E})=1-P(E)=1-\frac{1}{k+1}=\frac{k}{k+1}
\end{equation}
\end{itemize}
By using the rule of total probability,
\[
P_n(k) = P(E)(P_n(k)|E) + P(\Bar{E})(P_n(k)|\Bar{E}) \] %
from equations \ref{eq:pe}, \ref{eq:min3}, \ref{eq:pebar} and
\ref{eq:patebar}, we get,  %
\[ P_n(k) = \frac{1}{k+1} \left(1-
\frac{\binom{n-3}{k+1}}{\binom{n}{k+1}}\right) + \frac{k}{k+1} 
P_n(k+1) \]
\end{proof}
The above recursion for the probability of success can be solved in
$O(n)$ computation time. 

\begin{algorithm}[h!]
\DontPrintSemicolon
  
  \KwInput{$n$}
  \KwOutput{$P(k_n),k_n$}
  \SetKwFunction{FgetThp}{get\_optimal\_probability}

  \SetKwProg{Fn}{Def}{:}{}
  \Fn{\FgetThp{}}{
    $P[n] \longleftarrow 0$\\
    $P(k_n) \longleftarrow 0$\\
    $k_n \longleftarrow 0$\\
         \For{$k \gets n-1$ down \KwTo $0$} 
    {
        $r = ((n-k-1)*(n-k-2)*(n-k-3))/(n*(n-1)*(n-2))$\\
        $P[k] = (1/(k+1))*(1-r) + (k/(k+1))*P_n[k+1]$\\
        \If{$P[k] > P(k_n)$}{
            $P(k_n) \longleftarrow P[k]$\\
            $k_n \longleftarrow k$
        }
    }
    \KwRet $P(k_n),k_n$
  }
\caption{Get Optimal Probability ($P(k_n)$)}
\label{alg:gettop3th}
\end{algorithm}

\subsection{Asymptotic analysis of $P_n$}

For asymptotic analysis of $P_n$, we have the following,
\begin{prop}
\label{prop:top3.2} The sequence of functions $\hat{P}_n(x) :=
P_n(\floor{nx})$ converges uniformly on $[0,1]$ with 
$P(1)=0$ to
\begin{gather*}
P(x) =  -3x{\ln}(x)+3x^2-\frac{x^3}{2}-\frac{5x}{2}
\end{gather*}
\end{prop}
\begin{proof}
From Proposition \ref{prop:top3.1} we know that,
\[P_n(k) = G_n(k) + H_n(k)P_n(k+1)\]
where $G_n(k) = \frac{1}{k+1} \left(1-
\frac{\binom{n-3}{k+1}}{\binom{n}{k+1}}\right)$ and $H_n(k) =
\frac{k}{k+1}$.\\
Using expansion of $\binom{n}{r}$, 
we have
\[\frac{\binom{n-3}{k+1}}{\binom{n}{k+1}}= \frac{(n-k-1)(n-k-2)(n-k-3)}{n(n-1)(n-2)}\]
From Theorem \ref{thm:3} with 
$h_n(x) := n(1-H_n(\floor{nx}))$ and
$g_n(x) := nG_n(\floor{nx})$ we have, 
\[ g_n(x) = n\left(\frac{1}{\floor{nx}+1} \left(1-
\frac{(n-\floor{nx}-1)(n-\floor{nx}-2)(n-\floor{nx}-3)}{n(n-1)(n-2)}\right)
\right) \] and
\[ h_n(x) =
n\left(1-\left(\frac{\floor{nx}}{\floor{nx}+1}\right)\right) =
\left(\frac{n}{\floor{nx}+1}\right) \] 
As $n\rightarrow \infty$, $g_n(x) \rightarrow g(x)$, $h_n(x)
\rightarrow h(x)$ and $\floor{nx} \approx nx$.
\begin{itemize}
\item Let us estimate the value of $g(x)$: 
\[ g(x) = \underset{n \rightarrow \infty}{\lim} \left(\frac{n}{nx+1}
\left(1- \frac{(n-nx-1)(n-nx-2)(n-nx-3)}{n(n-1)(n-2)}\right) \right)
\] 
Now let us consider $\underset{n \rightarrow \infty}{\lim}
\frac{n}{nx+1}$.
\begin{equation*}
\begin{split}
\frac{n}{nx+1} = \frac{n}{nx\left(1+\frac{1}{nx}\right)}
\end{split}
\end{equation*}

Using expansion for $\frac{1}{1+\epsilon}$ where $\epsilon =
\frac{1}{nx}$. So,
\begin{equation*}
\frac{n}{nx+1} =  \frac{n}{nx}  \left(1-\frac{1}{nx} + O(n^{-2}) \right)
\end{equation*}

\begin{equation}
\label{eq:nbynx}
\frac{n}{nx+1} =  \frac{1}{x} + O(n^{-1})
\end{equation}

Now for $\frac{(n-nx-1)(n-nx-2)(n-nx-3)}{n(n-1)(n-2)}$,
\begin{equation} \label{nck}
\begin{split}
\frac{(n-nx-1)(n-nx-2)(n-nx-3)}{n(n-1)(n-2)} &=
\frac{n^3(1-x-\frac{1}{n})(1-x-\frac{2}{n})(1-x-
\frac{3}{n})}{n^3(1-\frac{1}{n})(1-\frac{2}{n})} \\ %
&= \frac{1-x-\frac{1}{n}}{1-\frac{1}{n}} 
\frac{1-x-\frac{2}{n}}{1-\frac{2}{n}} \left(1-x-\frac{3}{n}\right)
\end{split}
\end{equation}

Using expansion for $\frac{1}{1+\epsilon}$ where $\epsilon =
\frac{1}{n}$. So,
\begin{equation*}
\frac{1-x-\frac{1}{n}}{1-\frac{1}{n}} = \left( 1-x-\frac{1}{n} \right)
\left( 1+\frac{1}{n} + O(n^{-2}) \right) = \left( 1-x + O(n^{-1})
\right)
\end{equation*}

similarly,
\begin{equation*}
\frac{1-x-\frac{2}{n}}{1-\frac{2}{n}} = \left( 1-x-\frac{2}{n} \right)
\left( 1+\frac{2}{n} + O(n^{-2}) \right) = \left( 1-x + O(n^{-1})
\right)
\end{equation*}

Substituting the above results in equation \ref{nck}, we get,
\begin{equation} \label{eq:nck2}
\begin{split}
\frac{(n-nx-1)(n-nx-2)(n-nx-3)}{n(n-1)(n-2)} &= \left( 1-x + O(n^{-1})
\right)\left( 1-x + O(n^{-1}) \right)\left( 1-x -\frac{3}{n} \right)
\\ %
&= (1-x)^3 + O(n^{-1})
\end{split}
\end{equation}

Finally combining above results from equations \ref{eq:nck2},
\ref{eq:nbynx} and $g(x)$, we get,
\begin{equation*}
\begin{split}
g(x) &= \underset{n \rightarrow \infty}{\lim} \left(\frac{1}{x} +
O(n^{-1}) \right) \left(1- \left((1-x)^3 + O(n^{-1}) \right)\right) \\
&= \underset{n \rightarrow \infty}{\lim} \frac{1}{x} \left(1- (1-x)^3
\right) + O(n^{-1})
\end{split}
\end{equation*}

As $n \rightarrow \infty$ then $O(n^{-1})$ and higher order term
vanishes, so,
\begin{equation} \label{eq:gxtop3}
g(x) = \frac{1}{x} \left(1- (1-x)^3 \right)
\end{equation}

\item Let us estimate the value of $h(x)$:
\[ h(x) = \underset{n \rightarrow
\infty}{\lim}\left(\frac{n}{nx+1}\right) \] 
from equation \ref{eq:nbynx}, we know that,
\begin{equation*}
\begin{split}
h(x) = \underset{n \rightarrow \infty}{\lim} \left( \frac{1}{x} +
O(n^{-1}) \right)
\end{split}
\end{equation*}

As $n \rightarrow \infty$ then $O(n^{-1})$ and higher order term
vanishes, so,
\begin{equation} \label{eq:hxtop3}
h(x) = \frac{1}{x}
\end{equation}
\end{itemize}
Using Theorem \ref{thm:3} and substituting $g(x)$ and $h(x)$ from
above equations \ref{eq:gxtop3}, \ref{eq:hxtop3}. We get,
\begin{equation} \label{eq:pdeshx}
P^\prime(x) = \frac{1}{x} P(x) - \frac{1}{x} \left(1- (1-x)^3 \right)
\end{equation}
Using the general solution of a first-order linear differential
equation. We get,
\begin{gather*}
P(x) = \mu(x)^{-1} \left[ \mu(1)  P(1) + \int_{x}^{1} \mu(t)  g(t)  dt  \right],\\
\mu(x) = \exp\left(-\int h(x) dx \right)
\end{gather*}
As $P(1)=0$, so
\begin{gather*}
P(x) = \mu(x)^{-1} \left[ \int_{x}^{1} \mu(t) g(t) dt \right],\\
\mu(x) = \exp\left(-\int h(x) dx \right), h(x)=\frac{1}{x},
g(x)=\frac{1}{x} \left(1- (1-x)^3 \right)
\end{gather*}
Lets calculate the value for $\mu(x)$:
\begin{gather*}
\mu(x) = \exp\left(-\int h(x) dx \right) = e^{\left(-\int \frac{1}{x}
dx \right)} = e^{-{\ln}(x)} = x^{-1}
\end{gather*}
Substituting value of $\mu(x)$ is above equation for $P(x)$,
\begin{gather*}
P(x) = x \left[ \int_{x}^{1} \frac{1}{t} \left(\frac{1}{t} \left(1-
(1-t)^3 \right) \right) dt \right]
\end{gather*}
As $1- (1-t)^3 = 1-(1-t^3-3t+3t^2) = t^3+3t-3t^2$,
\begin{equation*}
\begin{split}
P(x) &= x \left[ \int_{x}^{1} \frac{1}{t^2} \left(t^3+3t-3t^2 \right)
dt \right] \\ %
&= x \left[ \int_{x}^{1} \left(t-3+\frac{3}{t} \right) dt \right] \\
&= x \left[ \frac{t^2}{2}-3t+3{\ln}(t) \right]_x^1 \\ %
&= x\left[ \left( \frac{1^2}{2}-3(1)+3{\ln}(1) \right)- \left(
\frac{x^2}{2}-3x+3{\ln}(x) \right) \right] \\ %
&= x\left[ \frac{-5}{2}-\left( \frac{x^2}{2}-3x+3{\ln}(x) \right)
\right] \\ %
&= x\left[-3{\ln}(x)+3x-\frac{x^2}{2}-\frac{5}{2} \right] \\%
	&=
-3x{\ln}(x)+3x^2-\frac{x^3}{2}-\frac{5x}{2}
\end{split}
\end{equation*}
\end{proof}

\subsubsection{Optimal Threshold}

The optimal threshold is the time instant at which the probability of
success, $P(x)$, is maximized. We can calculate this by taking the
derivative of $P(x)$ and equating it to zero. By substituting the
value of $P(x)$ from Propositions \ref{prop:top3.2} in equation
\ref{eq:pdeshx}, We get,
\begin{equation*}
\begin{split}
P^\prime(x) &= \frac{1}{x} 
\left(-3x{\ln}(x)+3x^2-\frac{x^3}{2}-\frac{5x}{2} \right) -
\frac{1}{x} \left(1- (1-x)^3 \right) \\ %
&= \left(-3{\ln}(x)+3x-\frac{x^2}{2}-\frac{5}{2} \right) - \frac{1}{x}
\left(x^3+3x-3x^2 \right)\\ %
&= \left(-3{\ln}(x)+3x-\frac{x^2}{2}-\frac{5}{2} \right) -
\left(x^2+3-3x \right) \\  %
&= -3{\ln}(x)+3x-\frac{x^2}{2}-\frac{5}{2} - x^2-3+3x \\ %
&= -3{\ln}(x)+6x-\frac{3}{2}x^2-\frac{11}{2} \\
\end{split}
\end{equation*}
Equating it to $0$, we get,
\begin{equation} \label{eq:optix}
-3{\ln}(x)+6x-\frac{3}{2}x^2-\frac{11}{2} =0
\end{equation}
Root of the above equation \ref{eq:optix} is $x^* \approx 0.259$.\\ %
At $x^* \approx 0.259$,
\begin{gather*}
P(x^*) = 
-3(0.259){\ln}(0.259)+3(0.259)^2-\frac{(0.259)^2}{2}-\frac{5(0.259)}{2}
= 0.59
\end{gather*}
So the optimal stopping threshold is $x^* \approx 0.259$ with maximum
success probability as $P(x^*) \approx 0.59$. 

\subsection{Numerical Simulation}

Table \ref{tab:top3optth} shows the optimal threshold $k_n$, the
normalized threshold $k_n/n$, and the corresponding optimal success
probability $P(k_n)$ obtained using Algorithm \ref{alg:gettop3th} for
various values of $n$. The results illustrate the convergence of the
optimal stopping proportion and success probability as $n$ increases.
\begin{table}[h!]
    \centering
    \renewcommand{\arraystretch}{1.4} 
    \setlength{\tabcolsep}{10pt}      
    \begin{tabular}{ |c|c|c|c|c|c|c|c| } 
         \hline
         $\mathbf{n}$ & 10 & 100 & 1000 & 10000 & 100000 & 1000000 & 10000000 \\ 
         \hline
         $\mathbf{k_n}$ & 2 & 26 & 260 & 2599 & 25997 & 259971 & 2599716 \\ 
         \hline
         $\mathbf{\frac{k_n}{n}}$ & 0.2 & 0.26 & 0.26 & 0.2599 & 0.25997 & 0.259971 & 0.2599716 \\
         \hline
         $\mathbf{P(k_n)}$ & 0.6640 & 0.6008 & 0.5953 & 0.59479 & 0.59473 & 0.59473 & 0.59472 \\
         \hline
    \end{tabular}
    \caption{Optimal values $k_n$, $k_n/n$ and $P(k_n)$ for different values of $n$}
    \label{tab:top3optth}
\end{table}

Figure \ref{fig:top3psuccess} shows the theoretical success
probability $P(k)$ as a function of the stopping threshold $k$ for
$n=100$ in the top-3 secretary problem. The probability increases
initially, reaches its maximum value $P(k) \approx 0.6008$ at the
optimal threshold $k=26$, and then gradually decreases, illustrating
the existence of an optimal stopping rule.

\begin{figure}[h!]
    \centering
    \includegraphics[width = 15cm]{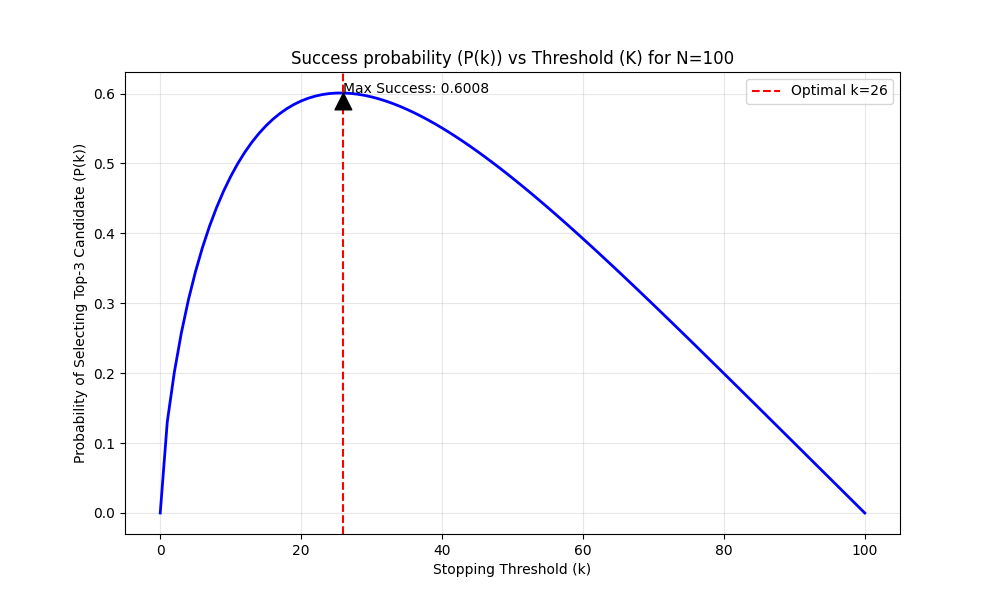}
    \caption{Success probability graph w.r.t $k$ at $n=100$}
    \label{fig:top3psuccess}
\end{figure}

Figure \ref{fig:top3simulation} presents the simulated success
probability based on $10,000$ trials for $n=100$ in the top-3
secretary problem. The simulation closely matches the theoretical
behavior, with the maximum success probability $(\approx0.6124)$
occurring near the optimal threshold $k=26$, validating the analytical
results.

\begin{figure}[h!]
    \centering
    \includegraphics[width = 15cm]{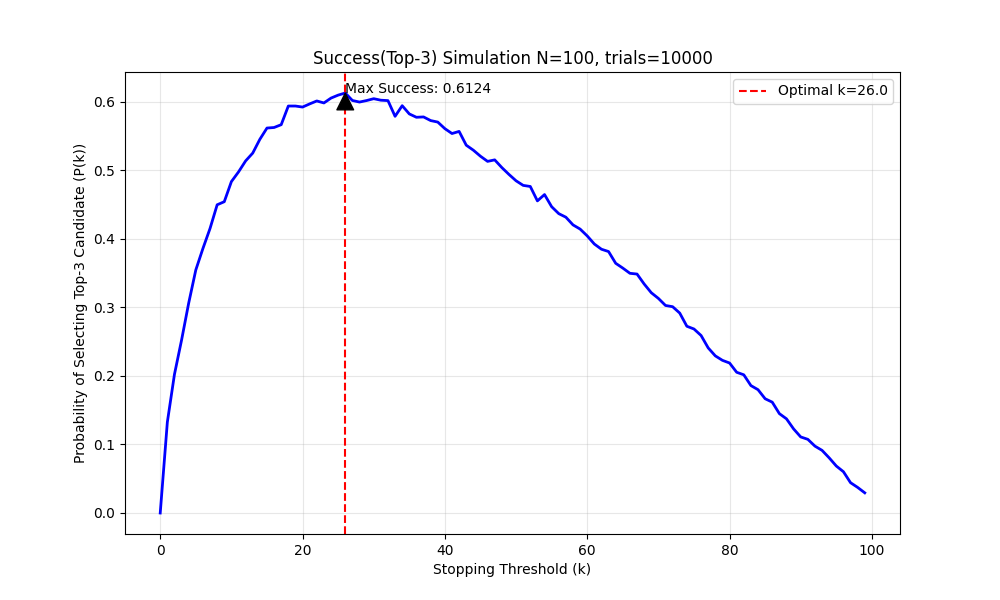}
    \caption{Simulation Success top-3 graph w.r.t $k$ at $n=100$, $trials=10000$}
    \label{fig:top3simulation}
\end{figure}

\section{Conclusion} \label{chap:conclusion} 

Here, we have studied two variations of the secretary problem.
Firstly, we studied a variation where each candidate may reappear a
second time with probability $p$. We proposed a threshold-based
strategy and derived a recursive formula to calculate the probability
of success at a general threshold $k$. We also provided an $O(n)$ time
dynamic programming algorithm to calculate the probability of success
at a general threshold $k$. We analyzed the extreme cases $p=0$ and
$p=1$, showing that the model reduces respectively to the classical
and returning secretary problem. The asymptotic analysis of the
proposed strategy is formulated in terms of differential equations,
which can be further studied to obtain optimal bounds.

Finally, 
we analyzed a relaxed variation of the classical secretary problem,
where selecting a candidate ranked among the top three globally is
considered a success. Using the optimal threshold structure of the
classical secretary problem, we derived a recursive formula and an
$O(n)$ time dynamic programming algorithm to calculate the probability
of success at a general threshold $k$. We further studied the
asymptotic behavior of the model and the optimal threshold and
corresponding optimal probability of success. 

\clearpage
\bibliographystyle{unsrt}  
\bibliography{references} 

\begin{thebibliography}{10}

\bibitem{Hill2009}
Theodore Hill.
\newblock Knowing when to stop.
\newblock {\em American Scientist}, 97(2):126, 2009.

\bibitem{SecProbWiki}
Wikipedia.
\newblock {Secretary problem}.
\newblock
  \url{http://en.wikipedia.org/w/index.php?title=Secretary\%20problem&oldid=987132965},
  2021.

\bibitem{vardi:LIPIcs:2015:4953}
Shai Vardi.
\newblock {The Returning Secretary}.
\newblock In Ernst~W. Mayr and Nicolas Ollinger, editors, {\em 32nd
  International Symposium on Theoretical Aspects of Computer Science (STACS
  2015)}, volume~30 of {\em Leibniz International Proceedings in Informatics
  (LIPIcs)}, pages 716--729, 2015.

\bibitem{bruss2000}
F.~Thomas Bruss.
\newblock Sum the odds to one and stop.
\newblock {\em Ann. Probab.}, 28(3):1384--1391, 06 2000.

\bibitem{dynkin}
E.~B. DYNKIN.
\newblock The optimum choice of the instant for stopping a markov process.
\newblock {\em Soviet Mathematics}, 4:627--629, 1963.

\bibitem{Babaioff2008}
Moshe Babaioff, Nicole Immorlica, David Kempe, and Robert Kleinberg.
\newblock Online auctions and generalized secretary problems.
\newblock {\em {ACM} {SIGecom} Exchanges}, 7(2):1--11, June 2008.

\bibitem{Babaioff2007}
Moshe Babaioff, Nicole Immorlica, David Kempe, and Robert Kleinberg.
\newblock A knapsack secretary problem with applications.
\newblock In {\em Approximation, Randomization, and Combinatorial Optimization.
  Algorithms and Techniques}, pages 16--28. Springer Berlin Heidelberg, 2007.

\bibitem{babaioff2007matroids}
Moshe Babaioff, Nicole Immorlica, and Robert Kleinberg.
\newblock Matroids, secretary problems, and online mechanisms.
\newblock In {\em Symposium on Discrete Algorithms (SODA'07)}, pages 434--443,
  January 2007.

\bibitem{Garrod2012jan}
Bryn Garrod, Grzegorz Kubicki, and Micha{\l} Morayne.
\newblock How to choose the best twins.
\newblock {\em {SIAM} Journal on Discrete Mathematics}, 26(1):384--398, January
  2012.

\bibitem{Vardi2014}
Shai Vardi.
\newblock The secretary returns.
\newblock {\em CoRR}, abs/1404.0614, 2014.

\bibitem{Ribas2018}
J.~M.~Grau Ribas.
\newblock A new look at the returning secretary problem.
\newblock {\em Journal of Combinatorial Optimization}, 37(4):1216--1236,
  September 2018.

\bibitem{Bayn2017}
L.~Bay{\'{o}}n, P.~Fortuny Ayuso, J.~M. Grau, A.~M. Oller-Marc{\'{e}}n, and
  M.~M. Ruiz.
\newblock The best-or-worst and the postdoc problems.
\newblock {\em Journal of Combinatorial Optimization}, 35(3):703--723, November
  2017.

\end{thebibliography}

\end{document}